\documentclass[runningheads]{llncs}
\usepackage[T1]{fontenc}
\usepackage{graphicx}
\usepackage{hyperref}
\usepackage{color}

\usepackage{verbatim}
\usepackage{multicol}
\usepackage{tikz}
\usepackage{subfigure}
\usepackage{amsmath}
\usepackage{amsfonts}
\usepackage{amssymb}
\usepackage{mathrsfs}

\usetikzlibrary{arrows,backgrounds,calc,fit,decorations.pathreplacing,decorations.markings,shapes.geometric, arrows}
\tikzset{every fit/.append style=text badly centered}
\definecolor{codegreen}{rgb}{0,0.6,0}
\definecolor{codegray}{rgb}{0.5,0.5,0.5}
\definecolor{codepurple}{rgb}{0.58,0,0.82}
\definecolor{backcolour}{rgb}{0.95,0.95,0.92}

\tikzstyle{internal} = [draw, fill, shape=circle]
\tikzstyle{external} = [shape=circle]
\tikzstyle{square}   = [draw, fill, rectangle]
\tikzstyle{triangle} = [draw, fill, regular polygon, regular polygon sides=3, inner sep=3pt]
\tikzstyle{pentagon} = [draw, fill, regular polygon, regular polygon sides=5, inner sep=2pt, minimum size=14pt]

\newcommand{\Rmnum}[1]{\expandafter\@slowromancap\romannumeral #1@}
\makeatother

\newcommand{\Holant}{\operatorname{Holant}}

\newcommand{\arity}{\operatorname{arity}}

\newcommand{\mf}[1]{\ensuremath{\mathcal{#1}}}

\newcommand{\la}{\langle}
\newcommand{\ra}{\rangle}

\newcommand{\tent}[2]{\ensuremath{#1 ^ {\otimes #2}}}

\newcommand{\wyy}[3]{\left[ \begin{smallmatrix} #1 \\ #2 \\ #3 \end{smallmatrix}\right]}
\newcommand{\xyy}[4]{\left[ \begin{smallmatrix} #1 \\ #2 \\ #3 \\ #4 \end{smallmatrix}\right]}

\begin{document}
\title{A combinatorial view of Holant problems on higher domains}

\author{Yin Liu}
\authorrunning{Y. Liu}
\institute{Google \\ Mountain View, CA 94043, USA \\
\email{hughliu@google.com}\\
}


%
\maketitle              
\begin{abstract}
On the Boolean domain, there is a class of symmetric signatures called ``Fibonacci gates''
for which a beautiful P-time combinatorial algorithm has been designed for the corresponding $\Holant$ problems. 

In this work, I give a combinatorial view for  $\Holant(\mf{F})$ problems on a domain of size 3 where $\mf{F}$ is a set of arity 3 functions with inputs taking values on the domain of size 3 and the functions share some common properties. The combinatorial view can also be extended to the domain of size 4.

Specifically, I extend the definition of ``Fibonacci gates'' to the domain of size 3 and the domain of size 4. Moreover, I give the corresponding combinatorial algorithms.

\end{abstract}
\section{Introduction and background}
In a lot of fields in 
computer science, machine learning and statistical physics, counting problems play a role. 
Holant problems encompass a broad class of
counting problems \cite{Backens21,BackensG20,CaiGW16,CaiL11,cai2009holant,CaiLX10,GuoHLX11,GuoLV13,MKJYC,valiant2006accidental,valiant2008holographic,Xia11}. 
This framework extends
edge-coloring models~\cite{szegedy2007edge,szegedy2010edge} while the latter is when the constraint functions are symmetric.
These problems also extend counting constraint satisfaction problems (a.k.a. CSP). 
It was proved that some prototypical Holant problems including counting perfect matchings, cannot be expressed as vertex-coloring models known as graph homomorphisms~\cite{Freedman-Lovasz-Schrijver-2007,HellN04}. 
The complexity classification program of counting problems 
is to classify the computational complexity
of these problems.

A Holant problem on a domain of size  $D$ is defined on a graph $G=(V, E)$ where $V, E$ represent the set of vertices and edges respectively. 
In our context, edges are variables and vertices are constraints.
Given a set of  constraint functions $\mathcal{F}$ defined on $D$,
a \emph{signature grid} $\Omega=(G, \pi)$  assigns to each vertex $v \in V$ a function (a.k.a. signature) $f_{v} \in \mathcal{F}$.
The goal is to compute the following \emph{partition} function 
$$\Holant_\Omega = \sum_{\sigma: E \rightarrow D} \prod_{v \in V}f_v\left(\left.\sigma\right|_{E(v)}\right).$$
$\sigma$ is each edge assignment and the summation contains in total $|E|^D$ many assignments. $E(v)$ are edges incident to vertex $v$. 
The computational problem is denoted by $\Holant (\mf{F})$.
Specifically, on the Boolean domain, it is over all $\{0,1\}$-edge assignments. On the domain of  size 3, it is over all $\{R,G,B\}$-edge assignments, signifying three colors Red, Green and Blue. And for  size 4, it is over all $\{R,G,B,W\}$-edge assignments.
On the Boolean domain, if every vertex has the 
\textsc{Exact-One}
function (which evaluates to 1 if
exactly one incident edge is 1, and evaluates to  0 otherwise),
then the partition function gives
the number of perfect matchings.
As another example, on domain size $k$, if every vertex has the 
\textsc{All-Distinct}
function,
then the partition function calculates
the number of valid $k$-edge colorings.

A $\Holant^*$ problem is a class of $\Holant$ problems where we assume all the unary signatures are freely available. In other words, $\Holant^*$ considers function sets containing all unary functions. $\Holant^*(\mf{F}) = \Holant(\mf{F}\ \cup\ \mf{U})$ where we denote by $\mf{U}$ the set of all unary functions. 

A \emph{symmetric} signature is a function that is invariant under
any permutation of its variables.
The value of such a signature  depends only on the numbers of each domain assigned to its input variables (i.e., edges). 
The number of variables
is its arity; unary, binary, ternary
signatures have arities 1, 2, 3 respectively.
We denote a symmetric ternary signature $g$ on a domain of size 3 by a ``triangle'' consisting of 10 numbers:
\begin{center}
\begin{tabular}{c c c c c c c}
     &&& $g_{3,0,0}$ &&&  \\
     && $g_{2,1,0}$ && $g_{2,0,1}$ &&\\
   & $g_{1,2,0}$ && $g_{1,1,1}$ && $g_{1,0,2}$ & \\
   $g_{0,3,0}$ && $g_{0,2,1}$ && $g_{0,1,2}$ && $g_{0,0,3}$ 
\end{tabular}    
\end{center}
where $g_{i,j,k}$ is the value on inputs having $i$ Red, $j$ Green and $k$ Blue.
Similarly, we denote a symmetric ternary signature $g$ on a domain of size 4 by a ``tetrahedron'':

\begin{figure}[h!]
    \centering
    \includegraphics[scale=0.7]{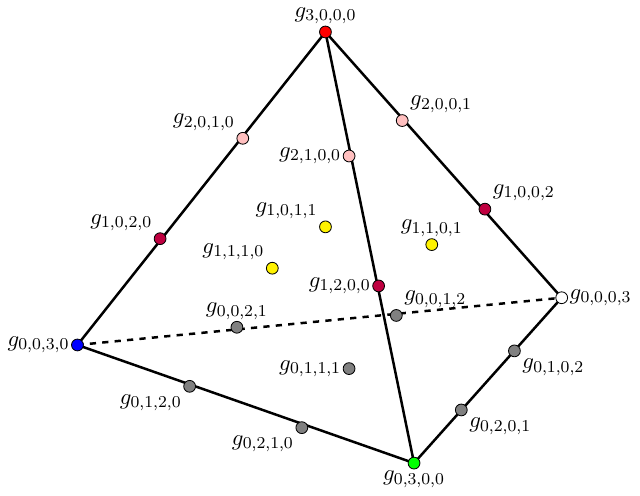}
    \label{pyramid}
\end{figure}

For the
classification of counting CSP, researchers have made a lot of progress~\cite{bulatov2006dichotomy}  
\cite{jacm/CaiC17,cai2016nonnegative,dyer2013effective}. Also, for Holant problems, \cite{CaiGW14,CaiLX13,Fu-Yang-Yin} are some of the milestone works. Specifically, in~\cite{Fibo}, ``Fibonacci gates''  is firstly proposed and studied on the Boolean domain. However, studying Holant problems on higher domains is particularly challenging. There are few existing works on higher domains. Among them, one
is \cite{CaiLX13}, in which a dichotomy for $\Holant^*(f)$ is proved where $f$ is a ternary complex symmetric function on domain size 3. 
Another piece of work is \cite{CaiGW14} which  investigates some higher domain  problems including $\kappa$-edge coloring  and derives some dichotomies. Another one is \cite{liu2023restricted} which proves some Holant dichotomies on domains 3 and 4 in restricted settings.  \\ 
Essentially, people only have scratched a little of higher domain classification programs and there are still many unknowns to explore.\\
\noindent \textbf{My contribution} In this paper, I focus on problems on domain sizes 3 and 4 while this time I am more interested in problems with the ``Fibonacci'' feature. I explore the quintessential nature of such specific signatures. I define `` \emph{Generalized} Fibonacci Gates'' on higher domains (Definitions~\ref{db3} and ~\ref{db4}), show that those problems are tractable and deliver our combinatorial algorithms  (Theorems~\ref{dom3_fibo_theo} and ~\ref{dom4_fibo_theo}). With these new exploration and discovery, I conjecture that there are Fibonacci Gates and the associated P-time algorithms on all higher domains.

\section{Fibonacci Gates on a domain of size 3}
By an orthogonal holographic transformation \cite{valiant2008holographic}, a domain 3 ternary symmetric signature is tractable when it is of the form $g = \tent{\alpha}{3} + \tent{\beta}{3} + \tent{\gamma}{3}$ where $\alpha, \beta, \gamma \in \mathbb{C}^3$  are mutually orthogonal to each other (specifically, $Tg = a  \tent{e_1}{3}+ b\tent{e_2}{3} +c \tent{e_3}{3}$ is tractable, for some $a, b, c\in \mathbb{R}$ and $T$ is an orthogonal matrix). 
Viewing $\alpha = (\alpha_1, \alpha_2, \alpha_3)^T, \beta = (\beta_1, \beta_2, \beta_3)^T, \gamma = (\gamma_1, \gamma_2, \gamma_3)^T$, I show that we can w.l.o.g. find $i\in \{1,2,3\}$ such that $\alpha_i \beta_i \gamma_i \ne 0$. 
Otherwise, it's easy to prove that either one of $\alpha, \beta, \gamma$ is a 0 vector in which case the problem is reducible to a lower domain, or, with some permutation of the three domains, $\alpha = c (1,0,0)^T, c\ne 0$ and $\beta, \gamma$ both have the form of $(0, *, *) ^T$ and the original signature $g$ is domain separable, also reducible to a lower domain and we have the corresponding dichotomies. 

Therefore,  it suffices to consider a signature of the form $$g = p \tent{\wyy{1}{a}{b}}{3} + q \tent{\wyy{1}{c}{d}}{3} + r \tent{\wyy{1}{e}{f}}{3} $$
where \begin{equation} \label{eq_orth}
    \begin{cases}
    ac + bd = -1\\
    ae + bf = -1\\
    ce + df = -1
    \end{cases}
\end{equation}
And $g$ can be written as (here I denote by $g_{i,j,k}$ the value of $g$ when $i$ incident edges are assigned color RED and $j$ edges GREEN and $k$ edges BLUE.)
\begin{center}
\begin{tabular}{c c c c c c c}
     &&& $g_{3,0,0}$ &&&  \\
     && $g_{2,1,0}$ && $g_{2,0,1}$ &&\\
   & $g_{1,2,0}$ && $g_{1,1,1}$ && $g_{1,0,2}$ & \\
   $g_{0,3,0}$ && $g_{0,2,1}$ && $g_{0,1,2}$ && $g_{0,0,3}$ 
\end{tabular}    
\end{center}
where \begin{equation} \label{g300}
    \begin{cases} g_{3,0,0} = p + q + r  &
    g_{2,1,0} = pa+qc+re \\
    g_{2,0,1} = pb+qd+rf &
    g_{1,2,0} = pa^2+qc^2 + re^2 \\
    g_{1,1,1} = pab+qcd+ref &
    g_{1,0,2} = pb^2+qd^2+rf^2\\
    g_{0,3,0} = pa^3+qc^3+rf^3  &
    g_{0,2,1} = pa^2b+qc^2d+re^2f \\
   g_{0,1,2} = pab^2+qcd^2+ref^2 &
   g_{0,0,3} = pb^3+qd^3+rf^3
   \end{cases}
\end{equation}

Consider each \textbf{``medium-sized'' triangle} with depth 2 consisting of 6 numbers over the 10 signature values of $g$, i.e.,  
\begin{center}
\begin{tabular}{c c c c c c c}
     &&& $g_{i,j,k}$ &&&  \\
     && $g_{i-1,j+1,k}$ && $g_{i-1, j, k+1}$ &&\\
   & $g_{i-2,j+2, k}$ && $g_{i-2,j+1,k+1}$ && $g_{i-2,j,k+2}$ 
\end{tabular}    
\end{center}
There are in total 3 such triangles. It can be proved directly that there exist four parameters $s,x,y,t$ such that \begin{equation}\label{bjeq}
    \begin{cases}
        g_{i-2,j+2,k}  = g_{i,j,k} + s g_{i-1.j+1,k} + x g_{i-1, j, k+1}\\
        g_{i-2,j+1,k+1} = x g_{i-1,j+1,k} + y g_{i-1,j,k+1}\\
        g_{i-2,j,k+2} = g_{i,j,k} + y g_{i-1,j+1,k} + t g_{i-1,j,k+1}
    \end{cases}
\end{equation}
and $s,x,y,t$ are independent of $p,q,r$. Moreover, they satisfy $sy + xt + 1 = x^2 + y^2$. 


\vspace{.1in}

\noindent The coefficients could be written in a ``\textbf{triangular}'' way below: \\
\begin{tabular}{ccc}
    &$1$&\\
    $s$&&$x$
\end{tabular}
\quad\quad
\begin{tabular}{ccc}
    &$0$&\\
    $x$&&$y$
\end{tabular}
\quad\quad 
\begin{tabular}{ccc}
    &$1$&\\
    $y$&&$t$
\end{tabular}

\vspace{-.1in}

\noindent Simply speaking, for each \textbf{``medium-sized'' triangle} of the shape 
\begin{tabular}{c c c c c c c}
     &&& *&&&  \\
     && * && * &&\\
   & * && * && * &
\end{tabular}  
(6 signature values among the total of 10),
the three signature values of the bottom line have a linear relationship with other three values on the top through $s,x,y,t$. 

\noindent In fact, in the settings above, 
\begin{equation} \label{indi}
    \begin{cases}
        x = -bdf\\
        y = -ace\\
        s = ace + a + c + e\\
        t = bdf + b + d + f
    \end{cases}
\end{equation}

\noindent It's easy to see that \begin{eqnarray*}
      sy + xt  + 1 &=& - ace(ace + a + c + e)  - bdf(bdf+b+d+f) + 1 \\
     &=& -ace(ace -a(ce+df) -c(ae+bf) -e(ac+bd))  \\ && - (ac+bd)(ae+bf)(ce+df) \\
     &=& ace(2 ace+adf+ cbf+ebd)  + bdf(2bdf+bce+dae + fac) \\ && - (ac+bd)(ae+bf)(ce+df) \\
     &=& (ace)^2 + (bdf)^2 \\
     &=& x^2 + y^2.
\end{eqnarray*}
Similarly, we can verify that the equations of $x,y,s,t$ represented by $a,b,c,d,e,f$ satisfy the linear recurrence relation~\ref{bjeq}.

\noindent\textbf{Another important thing} is that if we find such $s,x,y,t$, we can then recover the orthogonal matrix $\begin{bmatrix} 1 & 1 & 1 \\ a & c & e \\ b & d & f\end{bmatrix}$ by calculating the corresponding $a,c,e$. 
In fact, from Equation~\ref{eq_orth} and Equation~\ref{indi}, we get that 
\begin{eqnarray*}
    x^2 &=& (bdf)^2 = (-1-ac)(-1-ae)(-1-ce) \\
        &=& -((ace)^2 + (a+c+e)ace + (ac + ae + ce) + 1) \\
        &=& - y^2 + (s+y)y - (ac+ae+ce) - 1
\end{eqnarray*} 
Let $X = -y = ace, Y = s + y = a + c + e, Z = -x^2 - y^2 + (s+y)y -1 = ac + ae + ce$, 
we know that $a,c,e$ are the three roots of equation $$ t^3 - Yt^2 + Zt -X = 0.$$ 

\begin{definition} \label{db3}
We call a domain 3 symmetric signature $g$ (arity $\ge 2$) a \emph{generalized Fibonacci gate} (with parameter $s,x,y,t$ where $sy + xt + 1 = x^2 + y^2$) if \begin{equation*}
    \begin{cases}
        g_{i-2,j+2,k}  = g_{i,j,k} + s g_{i-1.j+1,k} + x g_{i-1, j, k+1}\\
        g_{i-2,j+1,k+1} = x g_{i-1,j+1,k} + y g_{i-1,j,k+1}\\
        g_{i-2,j,k+2} = g_{i,j,k} + y g_{i-1,j+1,k} + t g_{i-1,j,k+1}
    \end{cases}
\end{equation*}

$\forall i$,  $2 \le i \le \arity(g)$.  Specifically, any unary signature is a Fibonacci gate.

A set of signatures $\mathcal{F}$ is called \emph{generalized Fibonacci} if for some $s,x,y,t \in \mathbb{C}, sy + xt + 1 = x^2 + y^2$, every signature in $\mathcal{F}$ is a generalized Fibonacci gate with parameters $s,x,y,t$.

\end{definition}

\begin{theorem} \label{dom3_fibo_theo}
On a domain of size 3, for any finite set of generalized Fibonacci gates $\mathcal{F}$, the Holant problem $\Holant(\mathcal{F})$ is computable in polynomial time.
\end{theorem}

\begin{proof}   
If $\Gamma_1, \Gamma_2, ..., \Gamma_k$ are the connected components of a graph $\Gamma$, then 
$$ \Holant_\Gamma =  \prod_{j=1}^{k}\Holant_{\Gamma_j}.$$

So we only need to consider connected graphs as inputs.

Suppose $\Gamma$ has $n$ nodes and $m$ edges. First we cut all the edges in $\Gamma$. 
A node with degree $d$ can be viewed as an $\mathcal{F}$-gate with $d$ dangling edges. 
Now step by step we merge two dangling edges into one regular edge in the original graph, until we recover $\Gamma$ after $m$ steps. 
We prove that all the intermediate $\mathcal{F}$-gates still have generalized Fibonacci signatures with the same parameters $s,x,y,t$, and at every step we can compute the intermediate signature in polynomial time.
After $m$ steps we get $\Gamma$ as an $\mathcal{F}$-gate with no dangling edge; the only value of its signature is the Holant value we want.
To carry this out, we only need to prove that it is true for one single step.
There are two cases, depending on whether the two dangling edges to be merged are in the same component or not.
\begin{figure}[h]
    \centering
    \includegraphics[scale=0.22]{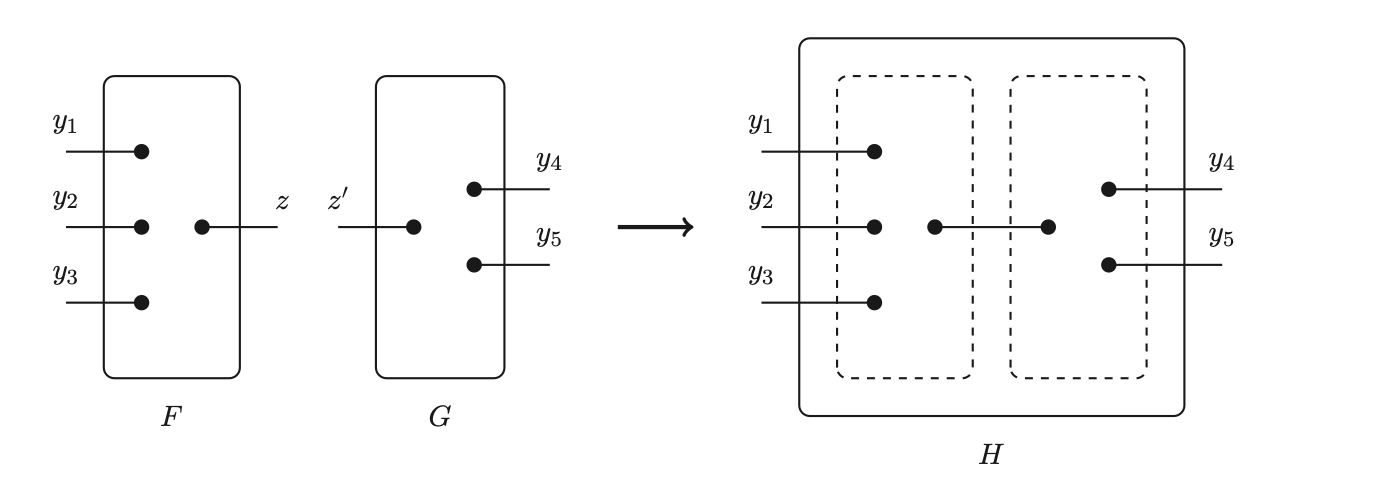}
    \caption{First operation}
    \label{fig_1}
\end{figure}
\begin{figure}[h]
    \centering
    \includegraphics[scale=0.2]{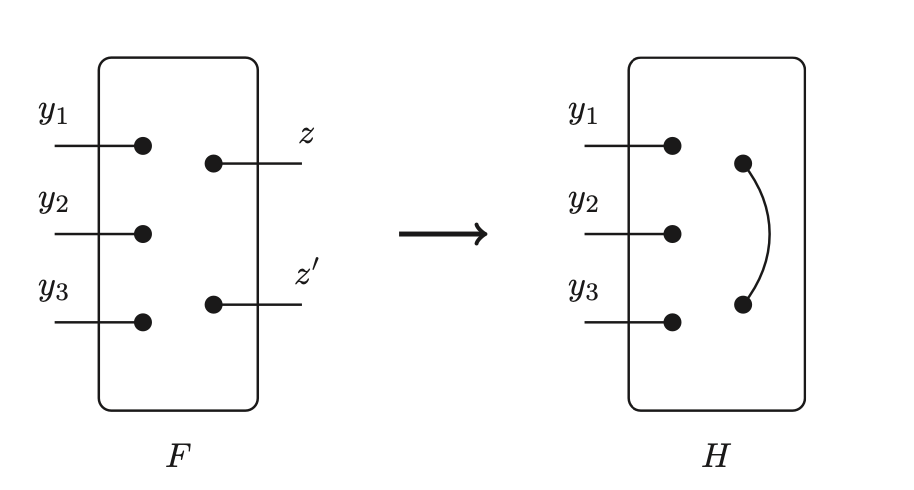}
    \caption{Second operation}
    \label{fig_2}
\end{figure}
These two operations are illustrated in Figures~\ref{fig_1} and ~\ref{fig_2}.

In the first case, the two dangling edges belong to two components before their merging (Figure~\ref{fig_1}). 
Let $F$ have dangling edges $y_1, y_2, ..., y_r, z$ and $G$ have dangling edges $y_{r+1},..., y_{r+w}, z'$. 
After merging $z$ and $z'$, we have a new $\mathcal{F}$-gate $H$ with dangling edges $y_1, ...,  y_{r+w}$. 
Inductively the signatures of $F$ and $G$ are both generalized Fibonacci gates with the same parameters $s,x,y,t$. 
We show that this remains so for the resulting $\mathcal{F}$-gate H. 

We first prove that $H$ is symmetric.
We only need to show that the value of $H$ is unchanged if the values of two inputs are exchanged.
Because $F$ and $G$ are symmetric, if both inputs are from $\{y_1,..., y_r\}$ or from $\{y_{r+1}, ..., y_{r+w}\}$, the value of $H$ is clearly unchanged. 
Suppose one input is from $\{y_1,..., y_r\}$ and the other is from $\{y_{r+1}, ..., y_{r+w}\}$.
By the symmetry of $F$ and $G$ we may assume these two inputs are $y_1$ and $y_{r+1}$. 
Thus we will fix an arbitrary assignment for $y_2, ..., y_r, y_{r+2}, ..., y_{r+w}$, 
and we want to show $H(u, y_2, ..., y_r, v, y_{r+2}, ..., y_{r+w}) = H(v, y_2, ..., y_r, u, y_{r+2}, ..., y_{r+w})$ where $u,v\in \{R,G,B\}$.

We will suppress the fixed values $y_2, .., y_r, y_{r+2}, ..., y_{r+w}$ and denote 
\vspace{-.05in}
$$ F_{uz} = F(u,y_2, ..., y_r,z)$$
$$G_{vz} = G(v, y_{r+2}, ..., y_{r+w}, z)$$
$$ H_{uv} = H(u, y_2, ..., y_r, v, y_{r+2}, ..., y_{r+w}).$$
Then by the definition of Holant, $H_{uv} = F_{uR}G_{vR} + F_{uG}G_{vG} + F_{uB}G_{vB}$, $u,v\in \{R,G,B\}$.
To make the notation simpler, we use subscript $\{1,2,3,4,5,6\}$ to represent $\{RR, RG, RB, GG,GB, BB\}$ respectively  now.
We also use $M_{ij}$ to denote $F_iG_j$.

Because $F$ is a generalized Fibonacci gate with parameters $s,x,y,t$, we have $F_4 = F_1 + sF_2 + x F_3$, $F_5 = xF_2 + yF_3$, and $F_6 = F_1 + yF_2 + tF_3$. Similar for $G$. 
Then we have the following \begin{eqnarray*}
H_{RG} & = & F_{RR}G_{GR} + F_{RG}G_{GG} + F_{RB}G_{GB} = F_1 G_2 + F_2 G_4 + F_3 G_5 \\
   & =  & F_1 G_2 + F_2 (G_1 + s G_2 + x G_3) + F_3 (x G_2 + y G_3)\\
   & = & M_{12} + M_{21} + sM_{22}+ xM_{23} + xM_{32} + yM_{33}
\end{eqnarray*}
\vspace{-.3in}
\begin{eqnarray*}
H_{GR} & =& F_{GR}G_{RR} + F_{GG} G_{RG} + F_{GB}G_{RB}  = F_2G_1 + F_4G_2 + F_5G_3 \\
    & =& F_2G_1 + (F_1 + sF_2 + xF_3)G_2 + (xF_2 +yF_3) G_3 \\
    & = & M_{12} + M_{21} + sM_{22} + xM_{23} + xM_{32} + yM_{33}
\end{eqnarray*} 
So $H_{RG} = H_{GR}$. Similarly we can prove $H_{uv} = H_{vu}$ for other $u,v \in \{R,G,B\}$. 
Hence $H$ is symmetric. We thus can use $1,2,3,4,5,6$ as subscripts for $H$ too.

Now we show that $H(y_1, ..., y_{r+w})$ is also a generalized Fibonacci gate with parameters $s,x,y,t$.
Since we have proved that $H$ is symmetric, we can choose any two input variables to prove it being Fibonacci. 
Again, we choose $y_1$ and $y_{r+1}$.
(This assumes that $y_1$ and $y_{r+1}$ exist, i.e., $F$ and $G$ are not unary functions. 
If either one of them is unary, the proof is just as easy.)
For any fixed values of all other variables, we have
\begin{eqnarray*}
    H_1  & = & H_{RR} = F_{RR}G_{RR} + F_{RG}G_{RG} + F_{RB}G_{RB}  = F_1G_1 + F_2G_2 + F_3G_3 \\
        & = & M_{11} + M_{22} + M_{33}
\end{eqnarray*}
\vspace{-.3in}
\begin{eqnarray*}
    H_2  & = & H_{RG} = F_{RR}G_{GR} + F_{RG}G_{GG} + F_{RB}G_{GB}  = F_1 G_2 + F_2 G_4 + F_3 G_5 \\
   & =  & F_1 G_2 + F_2 (G_1 + s G_2 + x G_3) + F_3 (x G_2 + y G_3)\\
   & = & sM_{22}+ yM_{33} + (M_{12} + M_{21}) +  x(M_{23} + M_{32}) 
\end{eqnarray*}
\vspace{-.3in}
\begin{eqnarray*}
    H_3 & = & H_{RB} = F_{RR}G_{BR} + F_{RG}G_{BG} + F_{RB}G_{BB} = F_1G_3 + F_2G_5 + F_3G_6 \\
    &=& F_1G_3 + F_2(xG_2 + yG_3) + F_3(G_1 + yG_2 + tG_3) \\
    &=& xM_{22} + t M_{33} + (M_{13} + M_{31}) + y( M_{23} + M_{32}) 
\end{eqnarray*}
\vspace{-.3in}
\begin{eqnarray*}
    H_4 & = & H_{GG} = F_{GR}G_{GR} + F_{GG}G_{GG} + F_{GB}G_{GB} = F_2G_2 + F_4G_4 + F_5G_5 \\
     &=& F_2G_2 + (F_1 + sF_2 + xF_3)(G_1 + sG_2 + xG_3)  + (xF_2 + yF_3)(xG_2 + yG_3)\\
     &=& M_{11} +  (s^2 + x^2 + 1)M_{22} + (x^2 + y^2)M_{33} \\
      & & + s (M_{12} + M_{21})  + x(M_{13} + M_{31}) + x(s+y)(M_{23} + M_{32})
\end{eqnarray*}
\vspace{-.3in}
\begin{eqnarray*}
    H_5 & = & H_{GB} = F_{GR} G_{BR}  + F_{GG}G_{BG} + F_{GB}G_{BB}  = F_2G_3 + F_4G_5 + F_5G_6 \\
     &=& F_2G_3 + (F_1+sF_2 + xF_3)(xG_2 + yG_3)  + (xF_2 + yF_3)(G_1 + yG_2 + tG_3) \\
     &=& x(s+y)M_{22} + y(x+t)M_{33} + x(M_{12} + M_{21})\\
     & &  + y(M_{13} +  M_{31}) + (1+sy+xt)M_{23} + (x^2+y^2)M_{32}\\
     &=& x(s+y)M_{22} + y(x+t)M_{33}  + x(M_{12} + M_{21}) + y(M_{13} +  M_{31}) \\ &&+ (x^2+y^2)(M_{23}+M_{32})
\end{eqnarray*}
(Here we use the condition that $x^2 + y^2  = sy + xt + 1$.)
\begin{eqnarray*}
    H_6 & = & H_{BB} =  F_{BR}G_{BR} + F_{BG}G_{BG}+F_{BB}G_{BB}  = F_3G_3 + F_5G_5 + F_6 G_6 \\
    &=& F_3G_3 + (xF_2 + yF_3)(xG_2 + yG_3) + (F_1 + yF_2 + tF_3)(G_1 + yG_2 + tG_3) \\
    &=& M_{11} + (x^2 + y^2) M_{22} + (y^2 + t^2 + 1) M_{33} \\
    & & + y(M_{12} + M_{21}) + t (M_{13} + M_{31}) + y(t+x) (M_{23} + M_{32})
\end{eqnarray*}

Using the relation of $s,x,y,t$ (i.e., $sy + xt + 1 = x^2 + y^2$), we can easily prove that $H_4 = H_1 + sH_2 + xH_3$, $H_5 = xH_2 + yH_3$, $H_6 = H_1 + yH_2 + tH_3$.

Next we consider the second case, where the two dangling edges to be merged are in the same component (Figure~\ref{fig_2}). 
Obviously, the signature for the new gate $H$ is symmetric. 
If $F$ below  is the signature before the merging operation, 
\begin{center}
\begin{tabular}{c c c c c c c c c}
     &&&& $F_1$ &&&&  \\
     &&& $F_2$ && $F_3$ &&&\\
   && $F_4$ && $F_5$ && $F_6$ && \\
   &$F_7$ && $F_8$ && $F_9$ && $F_{10}$ &\\
  $F_{11}$ && $F_{12}$ &&$F_{13}$ &&$F_{14}$ && $F_{15}$ 
\end{tabular} \\  
\vspace{.1in}
......
\end{center}
then the signature after the merging operation is $H$ below 
{\footnotesize
\begin{center}
\begin{tabular}{c c c c c}
     && $F_1 + F_4 + F_6$ &&  \\
     & $F_2 + F_7 + F_9$ && $F_3 + F_8 + F_{10}$ &\\
    $F_4 + F_{11} + F_{13}$ && $F_5 + F_{12} + F_{14}$ && $F_6 + F_{13} + F_{15}$ 
\end{tabular}   \\  \vspace{.1in}
...... 
\end{center}
}

Such an operation preserves the linear recurrence. 
It follows that $H$ is also a generalized Fibonacci gate with parameters $s,x,y,t$, and the signature has already been computed.

\end{proof}

\section{Fibonacci Gates on a domain of size 4}
We now move on to domain 4. Similarly, through an orthogonal transformation, we know that a ternary domain 4 symmetric signature is tractable when it is of the form $g = \tent{\alpha}{3} + \tent{\beta}{3} + \tent{\gamma}{3} + \tent{\delta}{3}$ where $\alpha, \beta, \gamma, \delta \in \mathbb{C}^3$ are mutually orthogonal to each other. 
Let $\alpha = (\alpha_1, \alpha_2, \alpha_3, \alpha_4)^T$, $\beta = (\beta_1, \beta_2, \beta_3, \beta_4)^T$, $\gamma = (\gamma_1, \gamma_2, \gamma_3, \gamma_4)^T$, $\delta = (\delta_1, \delta_2, \delta_3, \delta_4)^T$. 
Like on domain 3, we focus on the case that there is some $i\in \{1,2,3,4\}$ such that $\alpha_i\beta_i\gamma_i\delta_i \ne 0$.
By renaming the $i$-th domain to be the first domain, we can assume $\alpha_1\beta_1\gamma_1\delta_1 \ne 0$. 
Then it has the form of 
\begin{equation} \label{eq_d4g}
    g = q\xyy{1}{t_1}{t_2}{t_3}^{\otimes 3} 
+  r\xyy{1}{t_4}{t_5}{t_6}^{\otimes 3} 
+ s\xyy{1}{t_7}{t_8}{t_9}^{\otimes 3} 
+  t\xyy{1}{t_{10}}{t_{11}}{t_{12}}^{\otimes 3}
\end{equation}  where 
the four vectors are orthogonal to each other and $qrst\ne 0$.

\begin{figure}[h!] 
\centering
\begin{tikzpicture} [scale=0.7] 
\draw[very thick] (6.0, 0.5)--(0.3,2.5)--(4.5,7.8)--(8.5,3.3)--cycle;
\draw[very thick] (6.0, 0.5)--(4.5,7.8);
\draw[very thick, dashed] (0.3,2.5) -- (8.5,3.3);
\filldraw[fill=red] (4.5,7.8) circle(.1) node[anchor = south]{$g_{3,0,0,0}$} ;
\filldraw[fill=pink] (3.1,6.0) circle(.1) node[anchor = south east]{$g_{2,0,1,0}$} ;
\filldraw[fill=purple] (1.7,4.3) circle(.1) node[anchor = south east]{$g_{1,0,2,0}$} ;
\filldraw[fill=blue] (0.3,2.5) circle(.1) node[anchor = east]{$g_{0,0,3,0}$} ;
\filldraw[fill=pink] (4.9,5.7) circle(.1) node[anchor = east]{$g_{2,1,0,0}$} ;
\filldraw[fill=purple] (5.4,3.5) circle(.1) node[anchor=east]{$g_{1,2,0,0}$};
\filldraw[fill=green] (6.0,0.5) circle(.1) node[anchor=north] {$g_{0,3,0,0}$};
\filldraw[fill=pink] (5.8,6.3) circle(.1) node[anchor = south west] {$g_{2,0,0,1}$};
\filldraw[fill= purple] (7.2,4.8) circle(.1) node[anchor=south west] {$g_{1,0,0,2}$};
\filldraw[fill=white] (8.5,3.3) circle(.1) node[anchor = west] {$g_{0,0,0,3}$};
\filldraw[fill=gray] (3.0,2.8) circle(.1) node[anchor=south] {$g_{0,0,2,1}$};
\filldraw[fill=gray] (5.7, 3.0) circle(.1) node[anchor=south west] {$g_{0,0,1,2}$};
\filldraw[fill=gray] (2.2,1.8) circle(.1) node[anchor=north east] {$g_{0,1,2,0}$};
\filldraw[fill=gray] (4.1,1.1) circle(.1) node[anchor=north east] {$g_{0,2,1,0}$};
\filldraw[fill=gray] (6.8,1.4) circle(.1) node[anchor=north west] {$g_{0,2,0,1}$};
\filldraw[fill=gray] (7.7,2.4) circle(.1) node[anchor=north west] {$g_{0,1,0,2}$};

\filldraw[fill=yellow] (3.6,3.8) circle(.1) node[anchor=south east] {$g_{1,1,1,0}$};
\filldraw[fill=yellow] (6.3,4.2) circle(.1) node[anchor=south] {$g_{1,1,0,1}$};
\filldraw[fill=yellow] (4.5,4.5) circle(.1) node[anchor=south east] {$g_{1,0,1,1}$};
\filldraw[fill=gray] (4.9,2.1) circle(.1) node[anchor=east] {$g_{0,1,1,1}$};

\end{tikzpicture}
  \caption{$g$}
  \label{pmd}
\end{figure}

Figure~\ref{pmd} shows a domain 4 signature whose four domains are $\{R,G,B,W\}$ respectively. 
Here we denote by $g_{w,x,y,z}$ the value of $g$ when $w$ incident edges are assigned color RED, $x$ edges GREEN, $y$ edges BLUE and $z$ edges WHITE.
Similar to domain 3 signatures, if it is defined by Equation~\ref{eq_d4g},
then there exist  10 parameters $a,b,c,d,e,f,h,i,j,p$ such that in each \textbf{``medium-sized'' tetrahedron} (i.e., depth-2 tetrahedron) consisting of 10 numbers in ``3 layers'', \begin{equation}\label{eq_d4p}\begin{cases}
    g_{w-2,x+2,y,z} = g_{w,x,y,z} + a g_{w-1,x+1,y,z} + b g_{w-1,x,y+1,z} + c g_{w-1,x,y,z+1} \\
    g_{w-2,x,y+2,z} = g_{w,x,y,z} + d g_{w-1,x+1,y,z} + e g_{w-1,x,y+1,z} + f g_{w-1,x,y,z+1} \\
    g_{w-2,x,y,z+2} = g_{w,x,y,z} + h g_{w-1,x+1,y,z} + i g_{w-1,x,y+1,z} + j g_{w-1,x,y,z+1} \\
    g_{w-2,x+1,y+1,z} = b g_{w-1,x+1,y,z} + d g_{w-1,x,y+1,z} + p g_{w-1,x,y,z+1} \\
    g_{w-2,x+1,y,z+1} = c g_{w-1,x+1,y,z} + p g_{w-1,x,y+1,z} + h g_{w-1,x,y,z+1} \\
    g_{w-2,x,y+1,z+1} = p g_{w-1,x+1,y,z} + f g_{w-1,x,y+1,z} + i g_{w-1,x,y,z+1} 
\end{cases}
\end{equation}
and the 10 parameters also satisfy \begin{equation} \label{eq_d4q} \begin{cases}
    ad + be + cf + 1 = b^2 + d^2 + p^2 \\
    dh + ei + fj + 1 = f^2 + i^2 + p^2 \\
    ha + ib + jc + 1 = h^2 + c^2 + p^2 \\
    p^3 - (bi+cf+dh+1)p + bfh + cdi = 0
\end{cases}
 \end{equation}   

 \noindent The coefficients could be written as in Figure~\ref{fig:tet}. 
 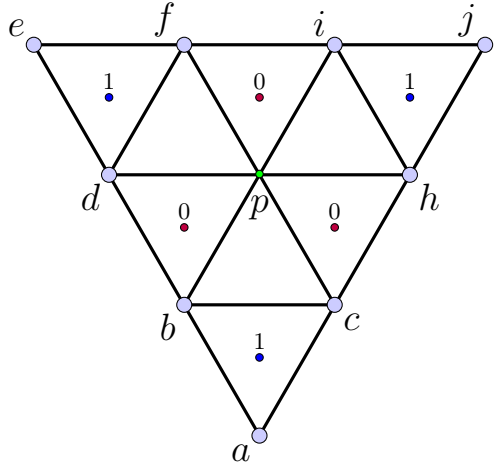
\begin{figure}[] 
\centering
  \begin{tikzpicture} [scale=1.0]

        \draw[very thick] (0,0) -- (6,0 ) -- (3,-5.2) -- cycle;
        \draw[very thick] (1,-1.73) -- (2,0) --(3, -1.73) -- (4,0)-- (5, -1.73) --(3, -1.73) -- (4,-3.46) --(2,-3.46)--(3, -1.73) -- cycle;
        
       \filldraw[fill = blue!20!white] (0,0) circle (1mm) node[anchor=south east] {\Large{$e$}};
        \filldraw[fill = blue!20!white] (2,0) circle (1mm) node[anchor=south east ] {\Large{$f$}};
       \filldraw[fill = blue!20!white] (4,0) circle (1mm) node[anchor=south east] {\Large{$i$}};
       \filldraw[fill = blue!20!white] (6,0) circle (1mm) node[anchor=south east] {\Large{$j$}};
       \filldraw[fill = blue!20!white] (1,-1.73) circle (1mm) node[anchor=north east] {\Large{$d$}};
       \filldraw[fill = blue!20!white] (2,-3.46) circle (1mm) node[anchor=north east] {\Large{$b$}};
       \filldraw[fill = blue!20!white] (3,-5.2) circle (1mm) node[anchor=north east] {\Large{$a$}};
       \filldraw[fill = blue!20!white] (4,-3.46) circle (1mm) node[anchor=north west] {\Large{$c$}};
       \filldraw[fill = blue!20!white] (5,-1.73) circle (1mm) node[anchor=north west] {\Large{$h$}};
       \filldraw[fill = green] (3,-1.72) circle (.5mm) node[anchor=north, inner ysep=8pt] {\Large{$p$}};

        \filldraw[fill=blue](1,-0.7) circle (.5mm) node[anchor=south] {1};
        \filldraw[fill=purple](3,-0.7) circle (.5mm) node[anchor=south] {0};
        \filldraw[fill=blue](5,-0.7) circle (.5mm) node[anchor=south] {1};
        \filldraw[fill=purple](2,-2.43) circle (.5mm) node[anchor=south] {0};
        \filldraw[fill=purple](4,-2.43) circle (.5mm) node[anchor=south] {0};
        \filldraw[fill=blue](3,-4.16) circle (.5mm) node[anchor=south] {1};
        
    \end{tikzpicture}
   \caption{Parameters of domain 4 fibonacci gate}
     \label{fig:tet}
\end{figure}
 It's \textbf{a top view of six tetrahedra} each consisting of 4 values. Three of them have a $1$ on top and three of them have a $0$ on top.  
 Each of these tetrahedra corresponds to a linear relationship mentioned in Equation~\ref{eq_d4p}.
 For example, the bottom tetrahedron is the one with a ``$1$'' on top and $a,b,c$ at the bottom. 
 This corresponds to the linear relation $g_{w-2,x+2,y,z} = g_{w,x,y,z} + a g_{w-1,x+1,y,z} + b g_{w-1,x,y+1,z} + c g_{w-1,x,y,z+1}$ in Equation~\ref{eq_d4p}. 
 Moreover, the ``dot product'' (i.e., sum of element-wise product) of every two different tetrahedra with a $1$ on top is equal to the ``dot product'' of the tetrahedron (with a $0$ on top) between the two  with itself. 
 For example, in Figure~\ref{fig:tet}, for tetrahedron indicated by $1,d,e,f$ and tetrahedron indicated by $1,h,i,j$, the tetrahedron between them is the one indicated by $0,p,f,i$, and $ \la \xyy{1}{d}{e}{f}, \xyy{1}{h}{y}{j} \ra = \la \xyy{0}{p}{f}{i}, \xyy{0}{p}{f}{i}  \ra $ is one of the first three equations in \ref{eq_d4q}. 
 Last but not least, there is a \textbf{cubic equation} regarding the center parameter $p$ in Equation~\ref{eq_d4q}, $p^3 - (bi+cf+dh+1)p + bfh + cdi = 0$. 

 It's worth mentioning that there are also another 3 quadratic relationships (Equation~\ref{eq_d4r}) among the parameters, however, with some calculation, they are incorporated in Equation~\ref{eq_d4q}. 
 \begin{equation} \label{eq_d4r} \begin{cases}
    ap + bf + ci = bc +  dp + ph  \\
    bh + di + pj = cp + pf + hi \\
    cd + pe + hf = bp + df + pi 
\end{cases}
 \end{equation} 
 In Figure~\ref{fig:tet}, basically Equation~\ref{eq_d4r} says that the ``dot product'' of each ``corner tetrahedron'' and the one in the opposite direction is equal to the ``dot product'' of the two ``side tetrahedra''. 
 For example, the first equation shows that tetrahedron $1,a,b,c$ dot producting with tetrahedron $0,p,f,i$ produces the same result as tetrahedron $0,b,d,p$ dot producting with $0,c,p,h$. 
 These 3 relations will be useful when we prove the combinatorial algorithm later.

 In this way, the 10 parameters have in total four constraints \ref{eq_d4q} and so the degree of freedom is 6 which is the same as that of a matrix whose four columns are orthogonal to each other written in $\begin{bmatrix} 
 1 & 1& 1& 1 \\
 t_1 & t_4 & t_7 & t_{10} \\
 t_2 & t_5 & t_8 & t_{11} \\
 t_3 & t_6 & t_9 & t_{12}
 \end{bmatrix} $.
With the 10 parameters satisfying Equation~\ref{eq_d4q}, we can form uniquely a $n$-arity domain 4 signature from the ``top'' 4 signature values, i.e., $g_{n,0,0,0}$, $g_{n-1,1,0,0}$, $g_{n-1,0,1,0}$, and $g_{n-1,0,0,1}$. 

Similar to domain 3, we have the following definition of Fibonacci gates and a corresponding tractable theorem. We leave the verbose proof in Appendix.

\begin{definition} \label{db4}
We call a domain 4 symmetric signature $g$ (arity $ \ge 2$) a \emph{generalized Fibonacci gate} (with parameters $a,b,c,d,e,f,h,i,j,p$ which satisfy Equation~\ref{eq_d4q}) if it satisfies Equation~\ref{eq_d4p} $\forall w, 2 \le w \le \arity(g)$. Any unary signature is also a Fibonacci gate. 

A set of signatures $\mathcal{F}$ is called \emph{generalized Fibonacci} if for some $a,b,c,d,e$, $f,h,i,j,p \in \mathbb{C}$ which satisfy Equation~\ref{eq_d4q}, every signature in $\mathcal{F}$ is a generalized Fibonacci gate with parameters $a,b,c,d,e,f,h,i,j,p$.

\end{definition}


\begin{theorem} \label{dom4_fibo_theo}
On a domain of size 4, for any finite set of generalized Fibonacci gates $\mathcal{F}$, the Holant problem $\Holant(\mathcal{F})$ is computable in polynomial time.
\end{theorem}

%
%

\begin{credits}
\subsubsection{\ackname} I sincerely thank my PhD advisor 
Jin-Yi Cai.
He really ought to be a coauthor for this paper for his contribution, but graciously declined. I am greatly indebted to his guidance and technical help. 
I also would like to thank 
Austen Z. Fan 
for his insightful discussion.

\end{credits}
%
%
%
%

\bibliographystyle{splncs04}
\bibliography{mybibfile}

\newpage
\section*{Appendix}

We give the proof of Theorem~\ref{dom4_fibo_theo} here.
\begin{proof}   
If $\Gamma_1, \Gamma_2, ..., \Gamma_k$ are the connected components of a graph $\Gamma$, then 
$$ \Holant_\Gamma =  \prod_{j=1}^{k}\Holant_{\Gamma_j}.$$

So we only need to consider connected graphs as inputs.

Suppose $\Gamma$ has $n$ nodes and $m$ edges. First we cut all the edges in $\Gamma$. 
A node with degree $d$ can be viewed as an $\mathcal{F}$-gate with $d$ dangling edges. 
Now step by step we merge two dangling edges into one regular edge in the original graph, until we recover $\Gamma$ after $m$ steps. 
We prove that all the intermediate $\mathcal{F}$-gates still have generalized Fibonacci signatures with the same parameters $a,b,c,d,e,f,h,i,j,p$, and at every step we can compute the intermediate signature in polynomial time.
After $m$ steps we get $\Gamma$ as an $\mathcal{F}$-gate with no dangling edge; the only value of its signature is the Holant value we want.
To carry this out, we only need to prove that it is true for one single step.
There are two cases, depending on whether the two dangling edges to be merged are in the same component or not.

These two operations are illustrated in Figures~\ref{fig_1} and ~\ref{fig_2}.

In the first case, the two dangling edges belong to two components before their merging (Figure~\ref{fig_1}). 
Let $F$ have dangling edges $y_1, y_2, ..., y_r, z$ and $G$ have dangling edges $y_{r+1},..., y_{r+w}, z'$. 
After merging $z$ and $z'$, we have a new $\mathcal{F}$-gate $H$ with dangling edges $y_1, ...,  y_{r+w}$. 
Inductively the signatures of $F$ and $G$ are both generalized Fibonacci gates with the same parameters $a,b,c,d,e,f,h,i,j,p$. 
We show that this remains so for the resulting $\mathcal{F}$-gate H. 

We first prove that $H$ is symmetric.
We only need to show that the value of $H$ is unchanged if the values of two inputs are exchanged.
Because $F$ and $G$ are symmetric, if both inputs are from $\{y_1,..., y_r\}$ or from $\{y_{r+1}, ..., y_{r+w}\}$, the value of $H$ is clearly unchanged. 
Suppose one input is from $\{y_1,..., y_r\}$ and the other is from $\{y_{r+1}, ..., y_{r+w}\}$.
By the symmetry of $F$ and $G$ we may assume these two inputs are $y_1$ and $y_{r+1}$. 
Thus we will fix an arbitrary assignment for $y_2, ..., y_r, y_{r+2}, ..., y_{r+w}$, 
and we want to show $H(u, y_2, ..., y_r, v, y_{r+2}, ..., y_{r+w}) = H(v, y_2, ..., y_r, u, y_{r+2}, ..., y_{r+w})$ where $u,v\in \{R,G,B,W\}$.

We will suppress the fixed values $y_2, .., y_r, y_{r+2}, ..., y_{r+w}$ and denote 
$$ F_{uz} = F(u,y_2, ..., y_r,z)$$
$$G_{vz} = G(v, y_{r+2}, ..., y_{r+w}, z)$$
$$ H_{uv} = H(u, y_2, ..., y_r, v, y_{r+2}, ..., y_{r+w}).$$
Then by the definition of Holant, $H_{uv} = F_{uR}G_{vR} + F_{uG}G_{vG} + F_{uB}G_{vB} + F_{uW}G_{vW}$, $u,v\in \{R,G,B,W\}$.
To make the notation simpler, we use subscript $\{0,1,2,3,4,5,6,7,8,9\}$ to represent $\{RR, RG, RB, RW, GG, GB, GW, BB, BW,\\ WW\}$  now.
We also use $M_{ij}$ to denote $F_iG_j$ now.

Because $F$ is a generalized Fibonacci gate with parameters $a,b,c,d,e,f,h$, $i,j,p$, we have \begin{eqnarray*}
    F_4 &= F_0 + &aF_1 + bF_2 + cF_3 \\
    F_7 &= F_0 + &dF_1 + eF_2 + fF_3 \\
    F_9 &= F_0 + &hF_1 + iF_2 + jF_3 \\
    F_5 &= &      bF_1 + dF_2 + pF_3 \\
    F_6 &= &      cF_1 + pF_2 + hF_3 \\
    F_8 &= &      pF_1 + fF_2 + iF_3 
\end{eqnarray*}Similar for $G$. 
Then we have the following \begin{eqnarray*}
H_{RG} & = & F_{RR}G_{GR} + F_{RG}G_{GG} + F_{RB}G_{GB} + F_{RW}G_{GW} \\
   & = & F_0 G_1 + F_1 G_4 + F_2 G_5  + F_3 G_6 \\
   & =  & F_0 G_1 + F_1 (G_0 + aG_1 + b G_2 + c G_3) \\
        & & + F_2 (b G_1 + d G_2 + p G_3) + F_3 (cG_1 + pG_2 + h G_3)\\
   & = & M_{01} + M_{10} + aM_{11} + dM_{22} + hM_{33} \\ & & + b(M_{12} + M_{21}) + c(M_{13} + M_{31}) + p (M_{23} + M_{32}) 
\end{eqnarray*}

\begin{eqnarray*}
H_{GR} & =& F_{GR}G_{RR} + F_{GG} G_{RG} + F_{GB}G_{RB} + F_{GW}G_{RW}\\
    & = & F_1G_0 + F_4G_1 + F_5G_2 + F_6G_3 \\
    & =& F_1 G_0 + (F_0 + aF_1 + b F_2 + c F_3)G_1 \\
        & & + (b F_1 + d F_2 + p F_3)G_2 + (cF_1 + pF_2 + h F_3)G_3\\
    & = & M_{01} + M_{10} + aM_{11} + dM_{22} + hM_{33} \\ & & + b(M_{12} + M_{21}) + c(M_{13} + M_{31}) + p (M_{23} + M_{32}) 
\end{eqnarray*} 
So $H_{RG} = H_{GR}$. Similarly we can prove $H_{uv} = H_{vu}$ for other $u,v \in \{R,G,B,W\}$. 
Hence $H$ is symmetric. We thus can use $0,1,2,3,4,5,6,7,8,9$ as subscripts for $H$ too.

Now we show that $H(y_1, ..., y_{r+w})$ is also a generalized Fibonacci gate with parameters $a,b,c,d,e,f,h,i,j,p$.
Since we have proved that $H$ is symmetric, we can choose any two input variables to prove it being Fibonacci. 
Again, we choose $y_1$ and $y_{r+1}$.
(This assumes that $y_1$ and $y_{r+1}$ exists, i.e., $F$ and $G$ are not unary functions. 
If either one of them is unary, the proof is just as easy.)
For any fixed values of all other variables, we have
\begin{eqnarray*}
   && H_0 = H_{RR} \\ & = & F_{RR}G_{RR} + F_{RG}G_{RG} + F_{RB}G_{RB} + F_{RW}G_{RW} \\
        & = & F_0G_0 + F_1G_1 + F_2G_2 + F_3G_3 \\
        & = & M_{00} + M_{11} + M_{22} + M_{33}
\end{eqnarray*}
$H_0$ is the evaluation of $H_{uv}$ where $u$ and $v$ both take color $R$. For the rest, similarly, the notations are introduced above and the equations below also explain that clearly.
\begin{eqnarray*}
    && H_1 = H_{RG} \\ & = & F_{RR}G_{GR} + F_{RG}G_{GG} + F_{RB}G_{GB} + F_{RW}G_{GW} \\
   & = & F_0 G_1 + F_1 G_4 + F_2 G_5  + F_3 G_6 \\
   & =  & F_0 G_1 + F_1 (G_0 + aG_1 + b G_2 + c G_3) \\
        & & + F_2 (b G_1 + d G_2 + p G_3) + F_3 (cG_1 + pG_2 + h G_3)\\
   & = &   aM_{11} + dM_{22} + hM_{33} + (M_{01} + M_{10}) \\ & & + b(M_{12} + M_{21}) + c(M_{13} + M_{31}) + p (M_{23} + M_{32}) 
\end{eqnarray*}
\begin{eqnarray*}
    && H_2 = H_{RB}\\ &=& F_{RR}G_{BR} + F{RG}G_{BG} + F{RB}G_{BB} + F_{RW}G_{BW}\\
    &=& F_0 G_2 + F_1 G_5 + F_2 G_7 + F_3 G_8 \\
    &=& F_0 G_2 + F_1 (bG_1 + d G_2 + p G_3)  \\ 
    & & + F_2 (G_0 + d G_1 + e G_2 + fG_3) + F_3 (pG_1 + fG_2 + i G_3) \\
    &=& bM_{11} + eM_{22} + i M_{33} + (M_{02} + M_{20}) \\
     & &   + d(M_{12} + M_{21}) + p(M_{13} + M_{31}) + f(M_{23} + M_{32})
\end{eqnarray*}
\begin{eqnarray*}
    && H_3 = H_{RW}\\ & = & F_{RR}G_{WR} + F_{RG}G_{WG} + F_{RB}G_{WB} + F_{RW}G_{WW} \\
    &=& F_0 G_3 + F_1 G_6 + F_2 G_8 + F_3 G_9\\
    &=& F_0 G_3 + F_1 (c G_1 + p G_2 + h G_3) \\ & & + F_2 (p G_1 + f G_2 + i G_3) + F_3 (G_0 + h G_1 + i G_2 + j G_3)\\
    &=& cM_{11} + fM_{22} + jM_{33} + (M_{03} + M_{30}) \\
    && + p(M_{12} + M_{21}) + h(M_{13} + M_{31}) + i(M_{23} + M_{32})
\end{eqnarray*}
\begin{eqnarray*}
    &&H_4 = H_{GG}\\ &=& F_{GR}G_{GR} + F_{GG}G_{GG} + F_{GB}G_{GB} + F_{GW} G_{GW}\\
     &=& F_1 G_1 + F_4 G_4 + F_5 G_5 + F_6 G_6 \\
     &=& F_1 G_1 + (F_0 + a F_1 + b F_2 + c F_3)(G_0 + a G_1 + b G_2 + c G_3) \\ 
     && + (b F_1 + d F_2 + p F_3)(b G_1 + d G_2 + p G_3) \\ & &+ (c F_1 + p F_2 + h F_3)(c G_1 + p G_2 + h G_3) \\
     &=& M_{00} + (1+a^2 + b^2 + c^2) M_{11} \\ && + (b^2+d^2+p^2)M_{22} + (c^2 + p^2+h^2)M_{33} \\
      && + a(M_{01} + M_{10}) + b(M_{02} + M_{20}) + c(M_{03}+M_{30}) \\
      && + (ab + bd + cp)(M_{12} + M_{21}) + (ac + bp + ch)(M_{13} + M_{31}) \\ && + (bc + dp + ph)(M_{23} + M_{32})
\end{eqnarray*}
\begin{eqnarray*}
    &&H_5 = H_{GB}\\ &=& F_{GR} G_{BR}  + F_{GG}G_{BG} + F_{GB}G_{BB} + F_{GW}G_{BW} \\
     &=& F_1 G_2 + F_4 G_5 + F_5 G_7 + F_6 G_8 \\
     &=& F_1 G_2 + (F_0 + aF_1 + bF_2 + cF_3)(bG_1 + dG_2 + pG_3) \\
     && + (bF_1 + dF_2 + pF_3) (G_0 + dG_1 + eG_2 + fG_3) \\ && + (cF_1 + pF_2 + hF_3)(pG_1 + fG_2 + iG_3) \\
     &=& (ab + bd + cp)M_{11} + (bd+de+pf)M_{22} + (cp+pf+ih)M_{33} \\
     && + b(M_{01}+M_{10}) + d(M_{02} + M_{20}) + p(M_{03}+M_{30}) \\
     && + (1+ad+be+cf)M_{12} + (b^2 + d^2 + p^2)M_{21} \\
     && + (ap+bf+ci)M_{13} + (bc + pd + hp)M_{31} \\
     && + (bp+df+pi)M_{23} + (cd+pe + hf)M_{32} \\
     &=& (ab + bd + cp)M_{11} + (bd+de+pf)M_{22} + (cp+pf+ih)M_{33} \\
     && + b(M_{01}+M_{10}) + d(M_{02} + M_{20}) + p(M_{03}+M_{30}) \\
     && +(b^2 + d^2 + p^2)(M_{12} + M_{21}) + (bc+pd+hp)(M_{13}+M_{31}) \\ && + (bp+df+pi)(M_{23}+M_{32})
\end{eqnarray*}
Here in the expression of $H_5$, the last equivalence is by Equations~\ref{eq_d4q} and ~\ref{eq_d4r}.

\begin{eqnarray*}
    &&H_6 = H_{GW}\\ &=& F_{GR}G_{WR} + F{GG}G_{WG} + F_{GB}G_{WB} + F_{GW}G_{WW} \\
    &=& F_1G_3 + F_4G_6 + F_5G_8 + F_6G_9 \\
    &=& F_1G_3 + (F_0 + aF_1 + bF_2 + cF_3)(cG_1 + pG_2 +hG_3) \\
     && + (bF_1 + dF_2 + pF_3)(pG_1 + fG_2 + iG_3)\\ && + (cF_1 + pF_2 + hF_3)(G_0 + hG_1 + iG_2 + jG_3) \\
     &=& (ac+bp+ch)M_{11} + (bp+df+pi)M_{22} + (ch+pi+hj)M_{33} \\
     && + c(M_{01} + M_{10}) + p(M_{02} + M_{20}) + h(M_{03}+M_{30}) \\
     && + (ap + bf + ci) M_{12} + (bc + dp + ph) M_{21} \\
     && + (1 + ah + bi + cj) M_{13} + (c^2 + p^2 + h^2) M_{31} \\
     && + (bh + di + pj)M_{23} + (cp+pf+hi)M_{32} \\
     &=& (ac+bp+ch)M_{11} + (bp+df+pi)M_{22} + (ch+pi+hj)M_{33} \\
     && + c(M_{01} + M_{10}) + p(M_{02} + M_{20}) + h(M_{03}+M_{30}) \\
     && + (bc+dp +ph)(M_{12}+M_{21})  + (c^2 + p^2+h^2)(M_{13}+M_{31}) \\ && + (cp+pf+hi)(M_{23}+M_{32})
\end{eqnarray*}
Similarly, in the expression of $H_6$, the last equivalence is by Equations~\ref{eq_d4q} and ~\ref{eq_d4r}.
\vspace{-.1in}
\begin{eqnarray*}
    && H_7 = H_{BB}\\ &=&  F_{BR}G_{BR} + F_{BG}G_{BG}+F_{BB}G_{BB} + F_{BW}G_{BW} \\
    &=& F_2G_2 + F_5G_5 + F_7G_7 + F_8G_8 \\
    &=& F_2G_2 + (bF_1 + dF_2 + pF_3)(bG_1 + dG_2 + pG_3) \\ 
    && (F_0 + dF_1 + eF_2 + fF_3)(G_0 + dG_1 + eG_2 + fG_3) \\&& + (pF_1+fF_2 + iF_3)(pG_1+fG_2+iG_3) \\
    &=& M_{00} + (b^2+d^2+p^2)M_{11} + (1+d^2+e^2+f^2)M_{22} \\
    && + d(M_{01}+M_{10}) + e(M_{02}+M_{20}) + f(M_{03}+M_{30}) \\
    && + (bd + de + fp)(M_{12}+M_{21})  + (bp+df+pi)(M_{13}+M_{31}) \\ && + (dp+ef+fi)(M_{23} + M_{32})  + (p^2+f^2+i^2)M_{33}
\end{eqnarray*}
\vspace{-.2in}
\begin{eqnarray*}
    && H_8 = H_{BW}\\ &=& F_{BR}G_{WR} + F_{BG}G_{WG} + F_{BB}G_{WB} + F_{BW}G_{WW} \\
    &=& F_2G_3 + F_5G_6 + F_7G_8 + F_8G_9 \\
    &=& F_2G_3 + (bF_1 + dF_2 + pF_3)(cG_1 + pG_2 + hG_3) \\ && + (F_0+dF_1 + eF_2 + fF_3)(pG_1+fG_2+iG_3)\\ 
            && + (pG_1 +fG_2+iG_3)(F_0+hF_1 + iF_2 + jF_3) \\
    &=& (bc+dp+ph)M_{11} + (dp+ef+fi)M_{22} + (ph+fi+ij)M_{33} \\
     && + p(M_{01} +M_{10}) + f(M_{02} + M_{20}) + i(M_{03}+M_{30}) \\
       && + (bp+df+pi)M_{12} + (cd + pe + hf)M_{21} \\
       && + (bh+di + pj)M_{13} + (cp + pf + hi)M_{31} \\
       && + (1+ dh+ei+fj)M_{23} + (p^2 + f^2 + i^2)M_{32} \\
    &=& (bc+dp+ph)M_{11} + (dp+ef+fi)M_{22} + (ph+fi+ij)M_{33} \\
     && + p(M_{01} +M_{10}) + f(M_{02} + M_{20}) + i(M_{03}+M_{30}) \\
       && +(bp+df+pi)(M_{12} + M_{21})  + (cp+pf+hi)(M_{13}+M_{31}) \\
       && + (p^2+f^2+i^2)(M_{32}+M_{23})
\end{eqnarray*} 
\vspace{-.2in}
Similarly, in the expression of $H_8$, the last equivalence is by Equations~\ref{eq_d4q} and ~\ref{eq_d4r}.
\vspace{-.03in}
\begin{eqnarray*}
    && H_9 = H_{WW}\\ &=& F_{WR}G_{WR} + F_{WG}G_{WG} + F_{WB}G_{WB} + F_{WW}G_{WW} \\
    &=& F_3G_3+F_6G_6+F_8G_8+F_9G_9 \\
    &=& F_3G_3 + (cF_1+pF_2 +hF_3)(cG_1+pG_2+hG_3) \\ && + (pF_1+fF_2 +iF_3)(pG_1+fG_2+iG_3) \\ 
          && + (F_0+hF_1+iF_2+jF_3)(G_0+hG_1+iG_2+jG_3)\\
    &=& M_{00} + (c^2+p^2+h^2)M_{11} + (p^2+f^2+i^2)M_{22} \\ && + (1 + h^2+i^2+j^2)M_{33} \\
    && + h(M_{01}+M_{10}) + i(M_{02}+M_{20}) + j(M_{03}+M_{30}) \\
    && + (cp+pf+hi)(M_{12}+M_{21}) + (ch+pi+hj)(M_{13}+M_{31}) \\&& + (ph+fi+ij)(M_{23}+M_{32}) 
\end{eqnarray*}

Consider $H_0+aH_1+bH_2+cH_3$, in fact, \vspace{-.05in} \begin{eqnarray*}
   & & H_0 + aH_1 + bH_2 + cH_3 \\
= & &  M_{00} + (1 + a^2 + b^2+c^2)M_{11} \\ &+&  (1+ad+be+cf)M_{22} + (1+ah+bi+cj)M_{33} \\
  &+ & a(M_{10}+M_{01})+b(M_{20}+M_{02}) + c(M_{03}+M_{30}) \\
  &+& (ab+bd+cp)(M_{12}+M_{21}) \\ &+&  (ac+bp+ch)(M_{13}+M_{31}) \\ &+&  (ap+bf+ci)(M_{23}+M_{32}) \\
= && M_{00} + (1 + a^2 + b^2+c^2)M_{11} \\ &+&  (b^2+d^2+p^2)M_{22} + (c^2+p^2+h^2)M_{33} \\
  &+ & a(M_{10}+M_{01})+b(M_{20}+M_{02}) + c(M_{03}+M_{30}) \\
  &+& (ab+bd+cp)(M_{12}+M_{21}) \\ &+&  (ac+bp+ch)(M_{13}+M_{31}) \\ &+&  (bc+dp+ph)(M_{23}+M_{32}) 
\end{eqnarray*}
where the last equivalence is by Equations~\ref{eq_d4q} and ~\ref{eq_d4r}, and the final value is just $H_4$

Similarly, we easily prove the linear relations between $H_5, H_6, \ldots, H_9$ and $H_0$, $H_1$, $H_2$, $H_3$ and conclude that $H$ also satisfies the relation~\ref{eq_d4p} with the same suite of parameters.

Next we consider the second case, where the two dangling edges to be merged are in the same component (Figure~\ref{fig_2}). 
Obviously, the signature for the new gate $H$ is symmetric. 
If $F$ in Figure~\ref{py2}  is the signature before the merging operation, then the signature after the merging operation is $H$ as in Figure~\ref{py3}. Note that both figures can be extended downward following the relations.

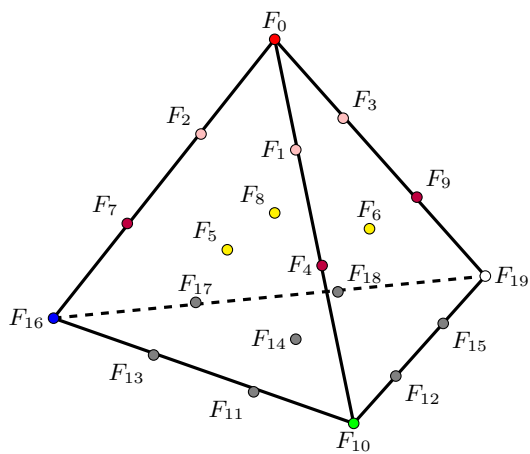
\begin{figure}[h] 
\centering
\begin{tikzpicture} [scale=0.7] 
\draw[very thick] (6.0, 0.5)--(0.3,2.5)--(4.5,7.8)--(8.5,3.3)--cycle;
\draw[very thick] (6.0, 0.5)--(4.5,7.8);
\draw[very thick, dashed] (0.3,2.5) -- (8.5,3.3);
\filldraw[fill=red] (4.5,7.8) circle(.1) node[anchor = south]{$F_0$} ;
\filldraw[fill=pink] (3.1,6.0) circle(.1) node[anchor = south east]{$F_2$} ;
\filldraw[fill=purple] (1.7,4.3) circle(.1) node[anchor = south east]{$F_7$} ;
\filldraw[fill=blue] (0.3,2.5) circle(.1) node[anchor = east]{$F_{16}$} ;
\filldraw[fill=pink] (4.9,5.7) circle(.1) node[anchor = east]{$F_1$} ;
\filldraw[fill=purple] (5.4,3.5) circle(.1) node[anchor=east]{$F_4$};
\filldraw[fill=green] (6.0,0.5) circle(.1) node[anchor=north] {$F_{10}$};
\filldraw[fill=pink] (5.8,6.3) circle(.1) node[anchor = south west] {$F_3$};
\filldraw[fill= purple] (7.2,4.8) circle(.1) node[anchor=south west] {$F_9$};
\filldraw[fill=white] (8.5,3.3) circle(.1) node[anchor = west] {$F_{19}$};
\filldraw[fill=gray] (3.0,2.8) circle(.1) node[anchor=south] {$F_{17}$};
\filldraw[fill=gray] (5.7, 3.0) circle(.1) node[anchor=south west] {$F_{18}$};
\filldraw[fill=gray] (2.2,1.8) circle(.1) node[anchor=north east] {$F_{13}$};
\filldraw[fill=gray] (4.1,1.1) circle(.1) node[anchor=north east] {$F_{11}$};
\filldraw[fill=gray] (6.8,1.4) circle(.1) node[anchor=north west] {$F_{12}$};
\filldraw[fill=gray] (7.7,2.4) circle(.1) node[anchor=north west] {$F_{15}$};

\filldraw[fill=yellow] (3.6,3.8) circle(.1) node[anchor=south east] {$F_5$};
\filldraw[fill=yellow] (6.3,4.2) circle(.1) node[anchor=south] {$F_6$};
\filldraw[fill=yellow] (4.5,4.5) circle(.1) node[anchor=south east] {$F_8$};
\filldraw[fill=gray] (4.9,2.1) circle(.1) node[anchor=east] {$F_{14}$};

\end{tikzpicture}

  \caption{Signature $F$. It can be extended downward (having more layers). }
  \label{py2}

\end{figure}

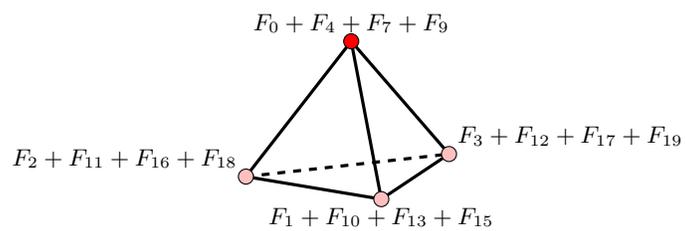
\begin{figure}[h] 
\centering
\begin{tikzpicture} [scale=1.0] 
\draw[very thick] (4.9,5.7)--(3.1,6.0)--(4.5,7.8)--(5.8,6.3)--cycle;
\draw[very thick] (4.9,5.7)--(4.5,7.8);
\draw[very thick, dashed] (3.1,6.0) -- (5.8,6.3);
\filldraw[fill=red] (4.5,7.8) circle(.1) node[anchor = south]{$F_0 + F_4 + F_7+F_9$} ;
\filldraw[fill=pink] (3.1,6.0) circle(.1) node[anchor = south east]{$F_2 + F_{11} + F_{16} + F_{18}$} ;
\filldraw[fill=pink] (4.9,5.7) circle(.1) node[anchor = north]{$F_1 + F_{10} + F_{13} + F_{15}$} ;
\filldraw[fill=pink] (5.8,6.3) circle(.1) node[anchor = south west] {$F_3 + F_{12} + F_{17} + F_{19}$};

\end{tikzpicture}

  \caption{Signature $H$. It can be extended downward (having more layers). }
  \label{py3}

\end{figure}

Such an operation preserves the linear recurrence. 
It follows that $H$ is also a generalized Fibonacci gate with parameters $a,b,c,d,e,f,h,i,j,p$, and the signature has already been computed.

\end{proof}

\end{document}